\documentclass[a4paper,twoside,11pt,notitlepage]{amsart}
\usepackage{a4wide}
\usepackage[T1]{fontenc}
\usepackage{amsthm}
\usepackage[latin1]{inputenc}
\usepackage{amsfonts,amsmath,epsfig,amscd,amssymb,latexsym}

\theoremstyle{plain}
\newtheorem{theorem}{Theorem}
\newtheorem{lemma}[theorem]{Lemma}

\theoremstyle{definition}
\newtheorem{definition}[theorem]{Definition}

\title{On a conjecture by Belfiore and Sol\'e on some lattices}
\author{Anne-Maria Ernvall-Hyt\"onen}
\address{Department of Mathematics, University of Helsinki}
\thanks{The funding from the Academy of Finland (grant number 138337) is gratefully acknowledged. The author would also like to thank Professor P\"ar Kurlberg and Dr. Camilla Hollanti for useful discussions, and Professors Patrick Sol\'e and Jean-Claude Belfiore for important comments.}
\begin{document}
\maketitle
\begin{abstract}
The point of this note is to prove that the secrecy function attains its maximum at $y=1$ on all known extremal even unimodular lattices. This is a special case of a conjecture by Belfiore and Sol\'e. Further, we will give a very simple method to verify or disprove the conjecture on any given  unimodular lattice. 
\end{abstract}
\section{Introduction}
Belfiore and Oggier defined in \cite{belfioggiswire} the secrecy gain
\[
\max_{y\in \mathbb{R}, 0<y}\frac{\Theta_{\mathbb{Z}^n}(yi)}{\Theta_{\Lambda}(yi)},
\]
where
\[
\Theta_{\Lambda}(z)=\sum_{x\in \Lambda}e^{\pi i ||x||^2z},
\]
as a new lattice invariant to measure how much confusion the eavesdropper will experience while the lattice $\Lambda$ is used in Gaussian wiretap coding. The function $\Xi_{\Lambda}(y)=\frac{\Theta_{\mathbb{Z}^n}(yi)}{\Theta_{\Lambda}(yi)}$ is called the secrecy function. Belfiore and Sol\'e then conjectured in \cite{belfisole} that the secrecy function attains its maximum at $y=1$, which would then be the value of the secrecy gain. The secrecy gain was further studied by Oggier, Sol\'e and Belfiore in \cite{belfisoleoggis}. The main point of this note is to prove the following theorem:

\begin{theorem}\label{main}
The secrecy function obtains its maximum at $y=1$ on all known even unimodular extremal lattices.
\end{theorem}

The method used here applies for any given even unimodular (and also for some unimodular but not even) lattices. This will be discussed in its own section. 
\section{Preliminaries and Lemmas}

For an excellent source on theta-functions, see e.g. the Chapter 10 in Stein's and Shakarchi's book \cite{steinshakarchi}. Define the theta function
\[
\Theta(z\mid\tau)=\sum_{n=-\infty}^{\infty}e^{\pi i n^2\tau}e^{2\pi i n z}.
\]
Now
\[
\Theta(z\mid \tau)=\prod_{n=1}^{\infty}(1-q^{2n})(1+q^{2n-1}e^{2\pi i z})(1+q^{2n-1}e^{-2\pi i z}),
\]
when $q=e^{\pi i \tau}$
The functions $\vartheta_2$, $\vartheta_3$ and $\vartheta_4$ can be written with the help of $\Theta(z\mid \tau)$ in the following way:
\begin{alignat*}{1}
\vartheta_2(\tau) & =e^{\pi i \tau/4}\Theta\left(\frac{\tau}{2}\mid\tau\right)\\
\vartheta_3(\tau) &=\Theta(0\mid \tau)\\
\vartheta_4(\tau) & =\Theta\left(\frac{1}{2}\mid \tau\right).
\end{alignat*}
Using the product representation for the $\Theta(z\mid \tau)$ function this reads
\begin{alignat*}{1}
\vartheta_2(\tau) & =e^{\pi i \tau/4}\prod_{n=1}^{\infty}(1-q^{2n})(1+q^{2n-1}e^{2\pi i\tau/2})(1+q^{2n-1}e^{-2\pi i \tau/2})\\ & =e^{\pi i \tau/4}\prod_{n=1}^{\infty}(1-q^{2n})(1+q^{2n})(1+q^{2n-2})\\
\vartheta_3(\tau) &=\prod_{n=1}^{\infty}(1-q^{2n})(1+q^{2n-1})(1+q^{2n-1})=\prod_{n=1}^{\infty}(1-q^{2n})(1+q^{2n-1})^2\\
\vartheta_4(\tau) & =\Theta\left(\frac{1}{2}\mid \tau\right)\prod_{n=1}^{\infty}(1-q^{2n})(1+q^{2n-1}e^{2\pi i /2})(1+q^{2n-1}e^{-2\pi i /2})\\ & =\prod_{n=1}^{\infty}(1-q^{2n})(1-q^{2n-1})^2.
\end{alignat*}

A lattice is called unimodular if its determinant $=\pm 1$, and the norms are integral, ie, $||x||^2\in \mathbb{Z}$ for all vectors $x$ on the lattice. Further, it is called even, if $||x||^2$ is even. Otherwise it is called odd. A lattice can be even unimodular only if the dimension is divisible by $8$. Odd unimodular lattices have no such restrictions.

Let us now have a brief look at the even unimodular lattices. Write the dimension $n=24m+8k$. Then the theta series of the lattice can be written as a polynomial of the Eisenstein series $E_4$ and the discriminant function $\Delta$:
\[
\Theta=E_4^{3m+k}+\sum_{j=1}^m b_i E^{3(m-j)+k}\Delta^j.
\]
Since $E_4=^\frac{1}{2}\left(\vartheta_2^8+\vartheta_3^8+\vartheta_4^8\right)$ and $\Delta=\frac{1}{256}\vartheta_2^8\vartheta_3^8\vartheta_4^8$, the theta function of an even unimodular lattice can be easily written as a polynomial of these basic theta functions. Furthermore, the secrecy function can be written as a simple rational function of $\frac{\vartheta_2^4\vartheta_4^4}{\vartheta_3^8}$:
\begin{multline}\label{polynomiksi}
\frac{\Theta_{\mathbb{Z}^n}}{\Theta_{\Lambda}}=\frac{\vartheta_3^{n}}{E_4^{3m+k}+\sum_{j=1}^mE_4^{3(m-j)+k}\Delta^j}\\=\left(\frac{\frac{1}{2}\left(\vartheta_2^8+\vartheta_3^8+\vartheta_4^8\right)+\sum_{j=1}^m\frac{b_j}{256^j\cdot 2^{3(m-j)+k}}\left(\vartheta_2^8+\vartheta_3^8+\vartheta_4^8\right)^{3(m-j)+k}\vartheta_2^{8j}\vartheta_3^{8j}\vartheta_4^{8j}}{\vartheta_3^{3m+k}}\right)^{-1}\\=\left(\left(1-\frac{\vartheta_2^4\vartheta_4^4}{\vartheta_3^8}\right)^{3m+k}+\sum_{j=1}^m\frac{b_j}{256^j}\left(1-\frac{\vartheta_2^4\vartheta_4^4}{\vartheta_3^8}\right)^{3(m-j)+k}\cdot\left(\frac{\vartheta_2^4\vartheta_4^4}{\vartheta_3^8}\right)^{2j}\right)^{-1}.
\end{multline}
Hence, finding the maximum of the secrecy function is equivalent to finding the minimum of the denominator of the previous expression in the range of $\frac{\vartheta_2^4\vartheta_4^4}{\vartheta_3^8}$.

\begin{definition}
An even unimodular lattice is called extremal if the norm of the shortest vector on the lattice is $2m+2$.
\end{definition}

It is worth noticing that the definition of extremal has changed. Earlier (see e.g. \cite{conwayodlyzkosloane}), extremal meant that the shortest vector was of length $\left\lfloor\frac{n}{8}\right\rfloor+1$. With the earlier definition the highest dimensional self-dual extremal lattice is in dimension $24$ (see \cite{conwayodlyzkosloane}), while with the current definition there is a selfdual extremal lattice in dimension $80$ (for a construction, see \cite{bachocnebe}).

Let us now turn to general unimodular lattices, in particular odd ones. Write $n=8\mu+\nu$, where $n$ is the dimension of the lattice. Just like a bit earlier, the theta function of any unimodular lattice (regardless of whether it's even or odd), can be written as a polynomial (see e.g. (3) in \cite{conwayodlyzkosloane}):
\[
\Theta_{\Lambda}=\sum_{r=0}^{\mu}a_r\vartheta_3^{n-8r}\Delta_8^r,
\]
where $\Delta_8=\frac{1}{16}\vartheta_2^4\vartheta_4^4$. Hence
\begin{equation}\label{polynomiksi2}
\frac{\Theta_{\mathbb{Z}^n}}{\Theta_{\Lambda}}=\left(\sum_{r=0}^{\mu}\frac{a_r}{16^r}\frac{\vartheta_2^{4r}\vartheta_4^{4r}}{\vartheta_3^{8r}}\right)^{-1}.
\end{equation}
Again, to determine the maximum of the function, it suffices to consider the denominator polynomial in the range of $\frac{\vartheta_2\vartheta_4}{\vartheta_3}$.

The following lemma is easy, and follows from the basic properties of the theta functions. The proof is here for completeness.

\begin{lemma}
Let $y\in \mathbb{R}$. The function
\[
f(y)=\frac{\vartheta_4^4(yi)\vartheta_2^4(yi)}{\vartheta_3^8(yi)}
\]
has symmetry: $f(y)=f\left(\frac{1}{y}\right)$
\end{lemma}
\begin{proof}
The formulas (21) in \cite{conwaysloane} give the following:
\begin{alignat*}{1}
\vartheta_2\left(\frac{i}{y}\right)= & \vartheta_2\left(\frac{-1}{yi}\right)=\sqrt{y}\vartheta_4(yi)\\
\vartheta_3\left(\frac{i}{y}\right)= & \vartheta_3\left(\frac{-1}{yi}\right)=\sqrt{y}\vartheta_3(yi)\\
\vartheta_4\left(\frac{i}{y}\right)= & \vartheta_4\left(\frac{-1}{yi}\right)=\sqrt{y}\vartheta_2(yi).
\end{alignat*}
Now
\[
f(y)=\frac{\vartheta_2^4(yi)\vartheta_4^4(yi)}{\vartheta_3^8(yi)}=\frac{y^{-2}\vartheta_4^4\left(\frac{i}{y}\right)y^{-2}\vartheta_2^4\left(\frac{i}{y}\right)}{y^{-4}\vartheta\left(\frac{i}{y}\right)}=f\left(\frac{1}{y}\right),
\]
which was to be proved.
\end{proof}

We may now formulate a lemma that is crucial in the proof of the main theorem:
\begin{lemma}\label{rajoite}
Let $y\in \mathbb{R}$. The function
\[
\frac{\vartheta_4^4(yi)\vartheta_2^4(yi)}{\vartheta_3^8(yi)}
\]
attains its maximum when $y=1$. This maximum is $\frac{1}{4}$.
\end{lemma}
\begin{proof}

To shorten the notation, write $g=e^{-\pi y}$. Notice that when $y$ increases, $g$ decreases and vice versa. Using the product representations for the functions $\vartheta_2(yi)$, $\vartheta_3(yi)$ and $\vartheta_4(yi)$, we obtain
\begin{multline*}
\frac{\vartheta_2(yi)\vartheta_4(yi)}{\vartheta_3(yi)^2}\\=\frac{\left(g^{1/4}\prod_{n=1}^{\infty}(1-g^{2n})(1+g^{2n})(1+g^{2n-2})\right)\left(\prod_{n=1}^{\infty}(1-g^{2n})(1-g^{2n-1})^2\right)}{\left(\prod_{n=1}^{\infty}(1-g^{2n})(1+g^{2n-1})^2\right)^2}\\
=g^{1/4}\left(\prod_{n=1}^{\infty}(1+g^{2n})(1+g^{2n-2})\right)\left(\prod_{n=1}^{\infty}(1-g^{2n-1})^2\right)\left(\prod_{n=1}^{\infty}(1+g^{2n-1})^{-4}\right).
\end{multline*}
Now
\[
\prod_{n=1}^{\infty}(1+g^{2n-2})=2\prod_{n=1}^{\infty}(1+g^{2n})
\]
and
\[
\prod_{n=1}^{\infty}(1+g^{2n-1})^{-4}=\prod_{n=1}^{\infty}(1+(-g)^n)^4=\left(\prod_{n=1}^{\infty}(1+g^{2n})\right)^4\left(\prod_{n=1}^{\infty}(1-g^{2n-1})\right)^4.
\]
Combining all these pieces together, we obtain
\begin{multline*}
\frac{\vartheta_2(yi)\vartheta_4(yi)}{\vartheta_3(yi)^2}\\=2g^{1/4}\left(\prod_{n=1}^{\infty}(1+g^{2n})^2\right)\left(\prod_{n=1}^{\infty}(1-g^{2n-1})^2\right)\left(\prod_{n=1}^{\infty}(1+g^{2n})\right)^4\left(\prod_{n=1}^{\infty}(1-g^{2n-1})\right)^4\\=2g^{1/4}\left(\prod_{n=1}^{\infty}(1+g^{2n})^6\right)\left(\prod_{n=1}^{\infty}(1-g^{2n-1})^6\right)=2\left(g^{1/24}\prod_{n=1}^{\infty}(1+(-g)^n)\right)^6.
\end{multline*}
Since the factor $2$ is just a constant, it suffices to consider the function $g^{1/24}\prod_{n=1}^{\infty}(1+(-g)^n)$. To find the maximum, let us first differentiate the function:
\[
\frac{\partial}{\partial g}\left(g^{1/24}\prod_{n=1}^{\infty}(1+(-g)^n)\right)=\left(g^{1/24}\prod_{n=1}^{\infty}(1+(-g)^n)\right)\left(\frac{1}{24g}+\sum_{n=1}^{\infty}\frac{n(-1)^ng^{n-1}}{1+(-g)^n}\right).
\]
Since $g^{1/24}\prod_{n=1}^{\infty}(1+(-g)^n)$ is always positive, it suffices to analyze the part $\frac{1}{24g}+\sum_{n=1}^{\infty}\frac{n(-1)^ng^(n-1)}{1+(-g)^n}$ to find the maxima. We wish to prove that the derivate has only one zero, because if it has only one zero, then this zero has to be located at $y=1$ (because the original function has symmetry, and therefore, a zero in the point $y$ results in a zero in the point $\frac{1}{y}$ which has to be separate unless $y=1$. To show that the derivative has only one zero, let us consider the second derivative, or actually, the derivative of the part $\frac{1}{24g}+\sum_{n=1}^{\infty}\frac{n(-1)^ng^{n-1}}{1+(-g)^n}$. Now
\[
\frac{\partial}{\partial g}\left(\frac{1}{24g}+\sum_{n=1}^{\infty}\frac{n(-1)^ng^{n-1}}{1+(-g)^n}\right)=-\frac{1}{24g^2}+\sum_{n=1}^{\infty}\left(\frac{n(n-1)(-1)^ng^{n-2}}{1+(-g)^n}-\frac{n^2g^{2(n-1)}}{(1+(-g)^n)^2}\right).
\]
Now we wish to show that this is negative when $g\in(0,1)$. Let us first look at the term $-\frac{1}{24g^2}$ and the terms in the sum corresponding the values $n=1$ and $n=2$. Their sum is
\[
-\frac{1}{24g^2}-\frac{1}{(1-g)^2}+\frac{2-2g^2}{(1+g^2)^2}=\frac{-73g^6+98g^5-51g^4-92g^3+21g^2+2g-1}{24g^2(1-g)^2(1+g^2)^2}.
\]
The denominator is positive when $g\in (0,1)$, and the nominator has two real roots, which are both negative (approximately $g_1\approx-0.719566$ and $g_2\approx-0.196021$). On positive values of $g$, the nominator is always negative. In particular, the nominator is negative when $g\in (0,1)$.

Let us now consider the terms $n>2$, and show that the sum is negative. Since the original function has symmetry $y\rightarrow \frac{1}{y}$, and we are only considering the real values of the theta series, we may now limit ourselves to the interval $y\in [1,\infty)$, which means that $g\in (0,e^{-\pi}]$.  Let us now show that the sum of two consecutive terms where the first one corresponds an odd value of $n$, and the second one an even value of $n$ is negative. The sum looks like the following:
\[
-\frac{n(n-1)g^{n-2}}{1-g^n}-\frac{n^2g^{2(n-1)}}{(1-g^n)^2}+\frac{n(n+1)g^{n-1}}{1+g^{n+1}}-\frac{(n+1)^2g^{2n}}{(1+g^{n+1})^2}.
\]
Let us estimate this, and take a common factor:
\begin{multline*}
<g^{n-2}n\left(-\frac{n-1}{1-g^n}-\frac{ng^n}{(1-g^n)^2}+\frac{(n+1)g}{1+g^{n+1}}-\frac{(n+1)g^{n+2}}{(1+g^{n+1})^2}\right)\\=g^{n-2}n\left(-\frac{n-1+g^n}{(1-g^n)^2}+\frac{(n+1)g}{(1+g^{n+1})^2}\right)<g^{n-2}n\left(\frac{-(n-1)-g^n+(n+1)g}{(1+g^{n+1})^2}\right)<0,
\end{multline*}
when
\[
(n-1)+g^n>(n+1)g.
\]
Since $(n-1)+g^n>(n-1)$, and $(n+1)g\leq (n+1)e^{-\pi}<\frac{n+1}{10}<n-1$, when $n\geq 2$, this proves that the first derivative has only one zero. This zero is at $y=1$. Since the second derivative is negative, it means that this point is actually the maximum of the function. The maximum value is:
\[
\frac{\vartheta_2^4(i)\vartheta_4^4(i)}{\vartheta_3^8(i)}=\frac{1}{4}
\]\end{proof}
\section{Proof of the main theorem}
Let us first do on the lattice $E_8$ as a warm-up case.
We wish to show that
\begin{theorem}
\begin{equation}\label{tavoite}
\Xi_{E_8}(y)\leq \Xi_{E_8}(1).
\end{equation}
\end{theorem}
\begin{proof}
 Notice that
\begin{multline*}
\Xi_{E_8}(yi)=\left(\frac{1}{2}\left(\frac{\vartheta_2(yi)^8+\vartheta_3(yi)^8+\vartheta_4(yi)^8}{\vartheta_3(yi)^8}\right)\right)^{-1}\\=\left(\frac{1}{2}+\frac{\left(\vartheta_2^4(yi)+\vartheta_4^4(yi)\right)^2-2\vartheta_2^4(yi)\vartheta_4^4(yi)}{2\vartheta_3(yi)^8}\right)^{-1}=\left(1-\frac{\vartheta_2^4(yi)\vartheta_4^4(yi)}{\vartheta_3^8(yi)}\right)^{-1},
\end{multline*}
Therefore, to show that (\ref{tavoite}) holds, it suffices to show that
\[
\frac{\vartheta_2(yi)^4\vartheta_4(yi)^4}{\vartheta_3(yi)^8}\leq \frac{\vartheta_2(i)^4\vartheta_4(i)^4}{\vartheta_3(i)^8},\] which is equivalent to showing that 
\[
\frac{\vartheta_2(yi)\vartheta_4(yi)}{\vartheta_3(yi)^2}\leq \frac{\vartheta_2(i)\vartheta_4(i)}{\vartheta_3(i)^2},
\] 
which we have already done in Lemma \ref{rajoite}.
\end{proof}
Let us now concentrate on the other cases. Again, write $z=\frac{\vartheta_2^4\vartheta_4^4}{\vartheta_3^8}$. The following table gives the secrecy functions of all known extremal even unimodular lattices (notice that these are known only in dimensions $8-80$):\\[0.1cm]
\begin{tabular}{c|c}
dimension & $\Xi$ \\ \hline
$8$ & $\left(1-z\right)^{-1}$\\
$16$ & $\left((1-z)^{2}\right)^{-1}$ \\
$24$ & $\left((1-z)^3-\frac{45}{16}z^2\right)^{-1}$ \\
$32$ & $\left((1-z)^4-\frac{15}{4}(1-z)z^2\right)^{-1}$ \\
$40$ & $\left((1-z)^5-\frac{75}{16}(1-z)^2z^2\right)^{-1}$ \\
$48$ & $\left((1-z)^6-\frac{45}{8}(1-z)^3z^2+\frac{3915}{2048}z^4\right)^{-1}$\\
$56$ & $\left((1-z)^7-\frac{105}{16}(1-z)^4z^2+\frac{21735}{4096}(1-z)z^4\right)^{-1}$ \\
$64$ & $\left((1-z)^8-\frac{15}{2}(1-z)^5z^2+\frac{4905}{512}(1-z)^2z^4\right)^{-1}$ \\
$72$ & $\left((1-z)^9-\frac{135}{16}(1-z)^6z^2+\frac{60345}{4096}(1-z)^3z^4-\frac{53325}{32768}z^6\right)^{-1}$\\
$80$ & $\left((1-z)^{10}-\frac{75}{8}(1-z)^7z^2+\frac{42525}{2048}(1-z)^4z^4-\frac{202125}{32768}(1-z)z^6\right)^{-1}$
\end{tabular}\\[0.3cm]
It suffices to show that the first derivatives of the denominators are negative because then the denominator is decreasing, and the function is increasing and obtains its maximum at $z=\frac{1}{4}$.

In dimension $16$ the derivative is
\[
-2(1-z)<0.
\]
In dimension $24$ the derivative is
\[
-3(1-z)^2-\frac{45}{8}z^2<0.
\] 
In dimension $32$ the derivative is
\[
-4(1-z)^3-\frac{15}{2}z(1-z)+\frac{15}{4}z^2=-4(1-z)^3-\frac{15}{4}z(2-2z+z)=-4(1-z)^3-\frac{15}{4}(2-z)<0.
\]
In dimension $40$ the derivative is
\[
-5(1-z)^4-\frac{75}{8}z(1-z)^2+\frac{75}{8}(1-z)z^2=-5(1-z)^4-\frac{75}{8}z(1-z)(2z-1)<0.
\]
In dimension $48$ the derivative is
\[
-6(1-z)^5-\frac{45}{4}(1-z)^3z+\frac{135}{8}(1-z)^2z^2+\frac{3915}{512}z^3<-6(1-z)^5+\frac{3915}{512}z^3<0
\]
In dimension $56$ the derivative is
\begin{multline*}
-7(1-z)^6-\frac{105}{8}z(1-z)^4+\frac{105}{4}(1-z)^3z^2+\frac{21735}{1024}z^3(1-z)-\frac{21735}{4096}z^4\\<-7(1-z)^6+\frac{21735}{1024}z^3(1-z)<0.
\end{multline*}
In dimension $64$ the derivative is
\[
-8(1-z)^7+\frac{75}{2}(1-z)^4z^2-15(1-z)^5z-\frac{4905}{256}z(1-z)^4+\frac{4905}{128}z^2(1-z)^3<0.
\]
In dimension $72$ the derivative is
\[
-9(1-z)^8+\frac{405}{8}(1-z)^5z^2-\frac{135}{8}(1-z)^6z+\frac{60345}{4096}\cdot 4(1-z)^3z^3-\frac{60345}{4096}\cdot 3 (1-z)^2z^4-6\cdot \frac{202125}{32768}z^5,
\]
which has real zeros $z_0\approx 0.3002$ and $z_1\approx 0.5222$, and when $z<z_0$, the derivative is negative, which proves this case.
It remains to consider the dimension 80, where the derivative is
\[
-10(1-z)^9-\frac{75}{4}(1-z)^7z+\frac{525}{8}(1-z)^6z^2+\frac{42525}{512}z^3(1-z)^4-\frac{42525}{512}(1-z)^3z^4+\frac{202125}{32768}z^6-\frac{1212750}{32768}(1-z)z^5,
\]
which has real roots $z_0\approx 0.2889$, $z_1\approx 0.4491$ and $z_2\approx 0.8620$, and is negative when $z<z_1$. This proves the theorem.

\section{Method for any given unimodular lattice}
Let $\Lambda$ be a  unimodular lattice. Then its secrecy function can be written as a polynomial $P(z)$, where $z=\frac{\vartheta_2^4\vartheta_4^4}{\vartheta_3^8}$ as shown in (\ref{polynomiksi}) and \ref{polynomiksi2}. Now, according Lemma \ref{rajoite}, $0\leq z\geq \frac{1}{4}$ (the lower bound does not follow from the lemma but from the fact that $z$ is a square of a real number). Therefore, it suffices to consider the polynomial $P(z)$ on the interval $[0,\frac{1}{4}]$. The conjecture is true if and only if the polynomial obtains its smallest value on the interval at $\frac{1}{4}$. Investigating the behaviour of a given polynomial to show whether one point is its minimum on a short interval is a very straightforward operation.

\end{document}